\documentclass[pdflatex,sn-mathphys-ay]{sn-jnl}

 \usepackage{url}
 \usepackage{graphicx}%
\usepackage{multirow}%
\usepackage{amsmath,amssymb,amsfonts}%
\usepackage{amsthm}%
\usepackage{mathrsfs}%
\usepackage[title]{appendix}%
\usepackage{xcolor}%
\usepackage{textcomp}%
\usepackage{manyfoot}%
\usepackage{booktabs}%
\usepackage{array}
\usepackage{algorithm}%
\usepackage{algorithmicx}%
\usepackage{algpseudocode}%
\usepackage{listings}%
\usepackage{verbatim}%
\usepackage{multicol}

\theoremstyle{thmstyleone}%
\newtheorem{theorem}{Theorem}
\newtheorem{proposition}[theorem]{Proposition}%

\newcommand{\ve}{{\bf e}}

\newcommand{\vC}{{\bf C}}
\newcommand{\vX}{{\bf X}}
\newcommand{\vJ}{{\bf J}}
\newcommand{\vI}{{\bf I}}

\newcommand{\vZ}{{\bf Z}}

\newcommand{\vTheta}{\mbox{\boldmath $\Theta$}}

\newcommand{\vvartheta}{\mbox{\boldmath $\vartheta$}}

\newcommand{\vSigma}{\mbox{\boldmath $\Sigma$}}

\newcommand{\valpha}{\mbox{\boldmath $\alpha$}}

\newcommand{\vtheta}{\mbox{\boldmath $\theta$}}

\newcommand{\dkonv}{\stackrel{d}{\rightarrow}}

\newcommand{\bqa}{\begin{eqnarray*}}
\newcommand{\eqa}{\end{eqnarray*}}
\newcommand{\bqan}{\begin{eqnarray}}
\newcommand{\eqan}{\end{eqnarray}}
\newcommand{\bit}{\begin{itemize}}
\newcommand{\eit}{\end{itemize}}
\newcommand{\ben}{\begin{enumerate}}
\newcommand{\een}{\end{enumerate}}
\newcommand{\beq}{\begin{equation}}
\newcommand{\eeq}{\end{equation}}
\newcommand{\bdes}{\begin{description}}
\newcommand{\edes}{\end{description}}

\newtheorem{THEOREM}{Theorem}
\newtheorem{LEMMA}{Lemma}

\theoremstyle{thmstyletwo}%
\newtheorem{example}{Example}%
\newtheorem{remark}{Remark}%

\theoremstyle{thmstylethree}%
\newtheorem{definition}{Definition}%

\begin{document}

\title[Consistent Variable Selection for GARCH-X Models]{Consistent Variable Selection for GARCH-X Models}


\author*{\fnm{Adriano Zanin} \sur{Zambom}}\email{adriano.zambom@csun.edu}

\author{\fnm{Beck} \sur{Saunders}}\email{beck.saunders.186@my.csun.edu}

\affil{\orgdiv{Department of Mathematics}, \orgname{California State University Northridge}, \orgaddress{\street{18111 Nordhoff Street}, \city{Northridge}, \postcode{91330}, \state{CA}, \country{USA}}}


\abstract{In this paper we develop a consistent variable selection procedure for GARCH-X models that identifies the truly relevant exogenous covariates influencing volatility dynamics. The proposed method is based on a multiple hypothesis testing framework with Wald-type test statistics and the Benjamini–Yekutieli False Discovery Rate (FDR) procedure to control the proportion of false discoveries. We establish the consistency of the selection rule, showing that it asymptotically recovers the correct set of covariates as the sample size increases. Monte Carlo simulations across different distributions and dependence structures validate the method’s accuracy and robustness. The procedure is applied to modeling the volatility of the S\&P 500 using macroeconomic and commodity indicators. }




\maketitle

\section{Introduction}\label{sec.intro}

The introduction of the Autoregressive Conditional Heteroscedasticity (ARCH) model by Engle in his seminal 1982 paper \citep{Engle1982} marked a significant advancement in time series econometrics, particularly for the modeling of financial asset volatility. Since its initial formulation, the ARCH framework has continued to serve as one of the principal analytical tool for capturing the time-varying nature of volatility in financial time series data. This enduring relevance is evidenced by numerous studies in the literature, from theory development to diverse applications \citep[e.g.,][]{BrennerEtAl1996, EnglePatton2001, MayMiguel2002, FabozziEtAl2014, BollerslevEtAl2016, BeyaztasEtAl2018, FrancqZakoian2022}. A significant generalization of the ARCH model was introduced by \cite{Bollerslev1986}, leading to the development of the Generalized Autoregressive Conditional Heteroskedasticity (GARCH) model, which allowed for a more flexible and parsimonious representation of volatility dynamics by incorporating lagged conditional variances.

While the foundational ARCH and GARCH models effectively address the endogenous dynamics of volatility, a growing body of empirical evidence suggests that the volatility of time series processes, especially in financial applications, is often systematically influenced by external, or exogenous, factors \citep{EnglePatton2001}. For instance, recent studies have increasingly highlighted the significant impact of variables such as news intensity on the volatility dynamics observed in stock markets and currency exchange rates \citep{sidorov2014garch, sadik2018news, chua2019information, LingZhu2014}. These studies underscore the importance of considering the role of external information in understanding the fluctuations in financial market uncertainty.

Although covariate effects on ARMA and GARCH models have received increased attention recently \citep[see, e.g.,][and references therein]{KristensenRahbek2005, NanaEtAl2013, GuoEtAl2014, HanKristensen2014, Han2015,FrancqSucarrat2017, DiopKengne2022}, the focus has been mostly on parameter estimation and their asymptotic properties~\citep{Engle2002, ChatterjeeDas2003, HanPark2008, HansenEtAl2012,  SucarratEtAl2016, ZambomGel2020}. Formal inference for  
a class of asymmetric GARCH models with multiple exogenous covariates was established in ~\cite{FrancqThieu2019}. Their approach allows one to perform hypothesis testing for the significance of the exogenous covariates using the asymptotic distribution of the estimators, which is shown to be a projection of the multivariate normal under some regularity conditions.

While some authors have addressed model selection in time series analysis (e.g., Bardet et al., 2020; Diop and Kengne, 2022; Kengne, 2023; Bardet et al., 2023), the available procedures primarily focus on determining the order of the model - such as identifying \(p\) and \(q\) in ARMA(\(p,q\)), GARCH(\(p,q\)) and some more general models.  The aforementioned approaches could, in principle, be extended to specifically address exogenous covariates, however they are not well suited for it. The available methods in the literature mostly rely on penalized model selection criteria that account for model complexity, akin to BIC-type approaches; however, exhaustively testing all possible combinations of exogenous variables with the penalized criteria becomes computationally prohibitive even for a moderate number of covariates.
At a conceptual level, model selection and hypothesis testing are intimately connected. In regression settings, selecting or excluding covariates can be viewed as performing hypothesis tests on the corresponding coefficients. Specifically, omitting variable \(j\) from the model is equivalent to not rejecting the null hypothesis \(H_0^j : \beta_j = 0\). This connection between model checking through hypothesis testing and variable selection has been extensively explored in regression contexts (see, for instance, \cite{AbramovichEtAl2006, KongXia2007, LiLiang2008, WangXia2009, HuangEtAl2010, StorlieEtAl2011, BuneaEtAl2006}).
Nevertheless, to date there is no practical procedure with theoretical results specifically devoted to variable selection among exogenous covariates in volatility models such as ARCH-X and GARCH-X.
To bridge the gap, the primary goal of this paper is to develop a variable selection method for exogenous covariates in GARCH-X models that recovers the covariates that compose the true underlying data-generating mechanism.

 In this paper we show that, based on the Benjamini and Yekuteli (2001) FDR procedure applied to p-values, a model checking procedure, in the context of a GARCH-X model, can be used to construct a consistent variable selection procedure by exploiting the aforementioned conceptual connection between model checking and variable selection. The consistency of the variable selection is formally established by showing that the selected index set is identical to the true set of relevant exogenous covariates with probability tenting to 1 as the sample size increases.
Simulations suggest that the proposed procedure consistently identifies the true set of relevant exogenous covariates with high probability as the sample size increases, offering a favorable trade-off between model complexity and accuracy across various distributional scenarios.

The remainder of the paper is organized as follows.  Section 2 introduces the test-based variable selection method and establishes its consistency. Section 3 presents simulation studies evaluating the finite-sample performance under different settings. Section 4 illustrates the method through an application to S\&P 500 volatility modeling. Section 5 concludes by summarizing the main contributions and emphasizing the practical relevance of the proposed method in achieving parsimonious and interpretable volatility models.

\section{Test Based Variable Selection for GARCH-X models}\label{sec.method}

Consider the GARCH(p,q)-X model

\begin{equation}\label{eqn:GARCHX}
\left\{
\begin{aligned}
\varepsilon_t &= \sigma_tw_t \\
\sigma_t^2(\vtheta) &= \omega + \sum_{i=1}^{p}\alpha_{i}\varepsilon_{t-i}^2 + \sum_{j=1}^{q}\beta_{j}\sigma_{t-j}^2 + \gamma^T\vX_{t-1}
\end{aligned}
\right.
\end{equation}
where $\vX_t = (X_{1,t}, X_{2,t}, \ldots, X_{d,t})$ is a vector of exogenous covariates, and the parameter vector can be denoted as 
\bqa
\vtheta := [\omega, \alpha_{1}, \ldots, \alpha_{p}, \beta_{1}, \ldots, \beta_{q}, \gamma_1, \ldots, \gamma_d]',
\eqa
which belongs to the parameter space $\vTheta \subseteq (0,\infty) \times [0, \infty]^{p+q+d}$.
The squared volatility $\sigma_t^2(\vtheta)$ is the best predictor of $\epsilon_t^2$ given an information set $\mathcal{F}_{t-1}$ available at time $t$, which defines the sigma-field generated by past returns $\{\epsilon_s, s < t\}$.
We assume that $E(\epsilon_t|\mathcal{F}_{t-1}) = \sigma_t^2(\vtheta) > 0$ and consequently that the covariates are almost surely positive and $\alpha_i \geq 0$, $\beta_j \geq 0$, $\gamma_k \geq 0$, and $\omega > 0$, $\forall 1 \leq i \leq p, \forall 1 \leq j \leq q, \forall 1 \leq k \leq d$.

Given data $\epsilon_1, \ldots, \epsilon_n$ and $\vX_{1}, \ldots, \vX_n$ (consider initial values $\epsilon_{1-\max(p,q)}, \ldots, \epsilon_0$, $\vX_0$ are also given), \cite{FrancqThieu2019} showed that the Gaussian QMLE $\hat{\vtheta}_n = \text{arg max}_{\vtheta \in \Theta}L_n(\vtheta)$   
is almost surely consistent for $\vtheta$ as the sample size increases and its asymptotic distribution under some regularity conditions is given by
\begin{align}
    \sqrt{n}(\hat{\vtheta}_n - \vtheta) \dkonv \vZ^\mathcal{C} \text{ as } n \to \infty,
\end{align}
where $\vZ \sim \mathcal{N}(0, \vSigma:=\vJ^{-1}\vI\vJ^{-1})$ and $\vZ^\mathcal{C}$ is the projection of the vector $\vZ$ on $\mathcal{C}$ defined by $\vZ^\mathcal{C} := \text{arg inf}_{\vC\in\mathcal{C}}\sqrt{(\vC - \vZ)^T\vJ(\vC - \vZ)}$,  $\mathcal{C} = \prod_{i=1}^{d}C_i$ is the local parameter space  such that $C_i = [0, \infty)$ if $\vtheta_{i} = 0$ and $\mathbb{R}$ if $\vtheta_{i} \neq 0$,
\begin{align}
    \textbf{J} &:= E\bigg[\frac{\partial^2 \ell_t(\vtheta)}{\partial \vtheta \partial \vtheta^T}\bigg] = E\bigg[\frac{1}{\sigma^4_t(\vtheta)}\frac{\partial \sigma^2_t(\vtheta)}{\partial \vtheta}\frac{\partial \sigma^2_t(\vtheta)}{\partial \vtheta^T}\bigg] \notag \\
    \textbf{I} &:= E\bigg[\bigg\{E[w^4_t | \mathscr{F}_{t-1}] - 1\bigg\}\frac{1}{\sigma^4_t(\vtheta)}\frac{\partial \sigma^2_t(\vtheta)}{\partial \vtheta}\frac{\partial \sigma^2_t(\vtheta)}{\partial \vtheta^T}\bigg],
    \notag
\end{align}
and $\ell_t(\vvartheta_0)$ is the log-QML function at time $t$.

Let $I_d = \{1, \ldots, d\}$ denote the set of indices of the $d$ available exogenous covariates and $I_0 = \{j_1, \ldots, j_{d_0}\} \subseteq I_d$ the (unknown) subset of indices corresponding to the $d_0$ truly relevant covariates, in other words, $I_0 = \{k \in I_d: \gamma_k \neq 0\}$. 
The objective of the proposed variable selection method is to recover the subset $I_0$ with high probability. Thus, we are interested in identifying the set of exogenous covariates with predictive significance for the volatility when it can depend on both endogenous and exogenous factors as in model (\ref{eqn:GARCHX}).

The proposed estimated set of indices of relevant covariates $\hat{I}$ is defined as follows. 
 Consider the hypothesis tests for the parameters corresponding to each exogenous covariate
\begin{equation} \label{Meth_H_test}
\begin{aligned}
    \text{H}_{0}^1 &: \gamma_1 = \ve_1^T\vtheta = 0 \;\;\;\;\;\; \text{ vs } \quad & \text{H}_{1_1}&: \gamma_1  > 0  ,\\
    \text{H}_{0}^2 &: \gamma_2 = \ve_2^T\vtheta = 0  \;\;\;\;\;\; \text{ vs } \quad & \text{H}_{1_2}&: \gamma_2 > 0,  \\
    & \hspace{15mm} \vdots \hspace{5mm} & & \hspace{5mm} \vdots \hspace{5mm} \\
    \text{H}_{0}^d &: \gamma_d = \ve_d^T\vtheta = 0  \;\;\;\;\;\; \text{ vs } \quad & \text{H}_{1_d}&: \gamma_d > 0,  \\
\end{aligned}
\end{equation}
where $\ve_k = [0,0,0, \ldots, 1, \ldots, 0,0,0]^T$ with 1 as the entry corresponding to the $k$-th parameter $\gamma_k$, that is, $\ve_k$ is the $(1 + p + q + k)$-th element of the canonical basis of $\mathbb{R}^{1+p+q+d}$.
The null hypothesis $H_0^k, k = 1, \ldots, d$ can be tested with the test statistic
\begin{align} \label{eq.t_test_stat}
    t_n(k) &= \sqrt{n}\frac{\ve_k^T\hat{\vtheta}_n}{\sqrt{\ve_k^T\hat{\vSigma}_n\ve_k}},
\end{align}
where $\hat{\vSigma}_n = \hat{\textbf{J}}^{-1}_n\hat{\textbf{I}}_n\hat{\textbf{J}}^{-1}_n$
and
\begin{align}
    \hat{\textbf{J}}_n &:= \frac{1}{n}\sum_{t=1}^{n}\frac{1}{\hat{\sigma}^2_t(\hat{\vtheta}_n)}\frac{\partial \hat{\sigma}^2_t(\hat{\vtheta}_n)}{\partial \vtheta}\frac{\partial \hat{\sigma}^2_t(\hat{\vtheta}_n)}{\partial \vtheta^T}
        \notag \\
        \hat{\textbf{I}}_n &:= \frac{1}{n}\sum_{t=1}^{n}(\hat{w_t}^4 - 1)\frac{1}{\hat{\sigma}^2_t(\hat{\vtheta}_n)}\frac{\partial \hat{\sigma}^2_t(\hat{\vtheta}_n)}{\partial \vtheta}\frac{\partial \hat{\sigma}^2_t(\hat{\vtheta}_n)}{\partial \vtheta^T},
\end{align}
with $\hat{w}_t = \varepsilon_t / \hat{\sigma}_t(\hat{\vtheta}_n)$ (see \cite{FrancqThieu2019} for more details). Here $\ve_k^T\hat{\vSigma}_n\ve_k$ estimates $\sigma_{kk} := \ve_k^T\vSigma\ve_k = \ve_k^T\vJ^{-1}\vI\vJ^{-1}\ve_k = \sum_{i=1}^d \sum_{j=1}^d (\mathbf{J}^{-1})_{k i} \, \mathbf{I}_{i j} \, (\mathbf{J}^{-1})_{j k}$, the variance of the estimator $\ve_k^T\hat{\vtheta}_n$. To show consistency of the variable selection method described below, we require the following condition.

(C1) Assume that $\sigma_{kk} \to 0$ and that $\sigma_{kk} < n/log(n)$ for all $k \in 1, \ldots, d$.

\noindent 
Condition (C1) ensures that the asymptotic variances of the estimators decrease sufficiently fast, while remaining well-controlled relative to the sample size. In particular, this restriction prevents excessive collinearity among the regressors and guarantees that the estimators retain adequate precision as $n$ increases. Consequently, (C1) regulates both the scale of the estimator's variance and its sensitivity to correlations between predictors (see \cite{BuneaEtAl2006}, \cite{Sarkar2023}, \cite{SarkarZhang2025} and references therein for further references on this topic).

Because the QMLE saturates the positivity constraint with probability 1/2 under $H_0^k$, the Wald test statistic $t_n^2(k)$ follows a chi-bar-square distribution (mix of Dirac at 0 and $\chi^2_1$). Consequently, the p-value for the test with rejection region $\{t_n^2(k) > \chi^2_1(1-2\alpha_n)\}$
with asymptotic level $\alpha_n$ is $\pi_k = (1/2)P(\chi_1^2 > t_n^2(k))$, or equivalently, $\pi_k = P(Z > |t_n(k)|)$.
Let $H_{0}^{(k)}$ denote the null hypothesis corresponding to the p-value $\pi_{(k)}$, $k = 1, \ldots, d$, where $\pi_{(1)} \leq \cdots \leq \pi_{(d)}$ denote the ordered p-values. We use the FDR procedure of \cite{BenjaminiHochberg1995} and \cite{BenjaminiYekutieli2001} to compute
\bqan\label{k.of.fdr}
\ell = \max_{i\in I_d}\left\{ k: \pi_{(k)} \leq \frac{k}{d}\frac{\alpha_n}{\sum_{l=1}^{d}l^{-1}}\right\},
\eqan
for a choice of the level $\alpha_n$, and reject the hypotheses $H_{0}^{(k)}$, $k = 1, \ldots, \ell$. If no such $\ell$ exists, no hypotheses are rejected. The proposed variable selection method selects the variables with indices corresponding to the $\ell$ rejected null hypotheses. Thus, $I_0$ is estimated by the set $\hat{I}$ of indices corresponding to the first $\ell$ ordered p-values.

Let $R$ represent the total number of rejected null hypotheses. Specifically, $R$ equals $\ell$ if such a value exists according to condition (\ref{k.of.fdr}), and $R$ is zero otherwise. Let $V$ denote the number of incorrectly rejected hypotheses. We define $Q$, the proportion of false rejections, as the ratio $V/R$ when $R$ is greater than zero, and as zero when no hypotheses are rejected. The false discovery rate (FDR) is then defined as the expected value of this proportion, $E(Q)$. \cite{BenjaminiYekutieli2001} showed that $E(Q) \leq \alpha_n(d-d_0)/d \leq \alpha_n$.

The variable selection procedure, and $\hat I$, are called consistent if $P(\hat{I} = I_0) \rightarrow 1$.
The consistency result presented here allows the significance of the predictors to be diminishing with an increasing sample size $n$. Lemma \ref{lemma:A1} establishes results for the p-values that are used in Lemma \ref{lemma:epsilon} and Theorem \ref{thm.consistency}. Lemma \ref{lemma:epsilon} shows that the smallest p-values from the hypotheses corresponding to the exogenous covariates are, with increasing probability, the p-values corresponding to the truly relevant ones, which is key to establishing the main result of consistency of the proposed method in Theorem \ref{thm.consistency}.

\begin{LEMMA} \label{lemma:A1} Let $t_n^2(k)$ and $\pi_k=(1/2)P(\chi^2_1 > t_n^2(k))$ be the test statistic and p-value for testing $H_0^k$. Assume that assumptions A1-A11 in \cite{FrancqThieu2019} and condition (C1) are satisfied. Then:\vspace{-3mm}
\ben \item[a)] For $k \notin I_0$ and any $\xi > 0$, we have $P(\pi_k \leq \xi) = \xi + o(1)$.
\item[b)] For $k \in I_0$, let $\xi_n> 0$, $n\ge 1$. If
\bqa
n^{1/2}\frac{\gamma_{k}}{\sigma_{kk}}\to \infty, \ \mbox{ and }\ \xi_n  = o\left(e^{-nc\frac{\gamma_{k}^2}{\sigma_{kk}^2}}\right) 
\eqa
for some constant $c$, then $P(\pi_k \geq \xi_n) = o(1)$.
\een
\end{LEMMA}
\begin{proof}
a) Let $\Psi_1(x)$ denote the cumulative distribution function of the $\chi^2_1$ distribution at $x$ and note that the p-value $\pi_k = (1/2)(1 - \Psi_1(t_n^2(k)))$. Under $H_0^k$, $k \notin I_0$ and
\bqa
P(\pi_k \leq \xi) &=& P\left(\frac{1}{2}(1 - \Psi_1(t_n^2(k))) \leq \xi\right) = P\left(\Psi_1(t_n^2(k)) \geq 1 - 2\xi\right)\\
&=&P\left(t_n^2(k) \geq \Psi_1^{-1}(1 - 2\xi)\right) \leq 2\xi + o(1),
\eqa
where the inequality follows from Collorary 1 in \cite{FrancqThieu2019} under $H_0^k$.

b) For $k \in I_0$ we have that $\gamma_k \neq 0$ so that

\bqa
P(\pi_k \geq \xi_n) &=& P\left(\frac{1}{2}(1 - \Psi_1(t_n^2(k))) \geq \xi\right) = P\left(1 - \Phi(|t_n(k)|) \geq \xi_n\right)\\
&=&P\left(|t_n(k)| \leq \Phi^{-1}(1-\xi_n)\right)\\
&=&P\left(\left|\sqrt{n}\frac{\ve_k^T\hat{\vtheta}_n}{\sqrt{\ve_k^T\hat{\vSigma}_n\ve_k}}\right| \leq \Phi^{-1}(1-\xi_n)\right)\\
&=&P\left(\left|\sqrt{n}\frac{\ve_k^T\hat{\vtheta}_n - \ve_k^T\vtheta_n}{\sqrt{\ve_k^T\hat{\vSigma}_n\ve_k}} + \sqrt{n}\frac{\ve_k^T\vtheta_n}{\sqrt{\ve_k^T\hat{\vSigma}_n\ve_k}}\right| \leq \Phi^{-1}(1-\xi_n)\right)\\
&\leq&\Phi\left(\Phi^{-1}\left(1 - \xi_n\right) - \frac{\sqrt{n}\gamma_k}{\sigma_{kk}}\right) - \Phi\left(-\Phi^{-1}\left(1 - \xi_n\right) - \frac{\sqrt{n}\gamma_k}{\sigma_{kk}}\right) + o(1)
\eqa
which converges to 0 by Theorem 2 and Proposition 1 in \cite{FrancqThieu2019} for the choice of  $\xi_n =  o\left(e^{-n\gamma_{k}^2/\sigma_{kk}^2}\right)$,
using the fact that $\Phi^{-1}(1 - \xi_n) = O(\sqrt{- \ln(\xi_n)})$ (\cite{Feller1971}).

\end{proof}

\begin{LEMMA} \label{lemma:epsilon} Let ${\cal E}_n$ denote the event where the smallest $d_0$ p-values are the p-values corresponding to the $d_0$ significant covariates, i.e.,  with $I_0 = \{j_1,\ldots,j_{d_0}\}$,
\bqa
{\cal E}_n = \big[\{\pi_{(1)},\ldots,\pi_{(d_0)}\} = \{\pi_{j_1},\ldots,\pi_{j_{d_0}}\}\big].
\eqa
Then, under the assumptions of Lemma \ref{lemma:A1}, we have that 
\bqa
\lim_{n\rightarrow \infty}P({\cal E}_n) = 1.
\eqa
\end{LEMMA}
\begin{proof}
Let $\xi$ be any number between 0 and 1, and write
\bqa
P({\cal E}_n^c) & \leq & \sum_{j \in I_0}\sum_{i \notin I_0}P(\pi_i < \pi_j) \nonumber \\
&=& \sum_{j \in I_0}\sum_{i \notin I_0}\big[P([\pi_i < \pi_j] \cap [\pi_i \leq \xi])  + P([\pi_i < \pi_j] \cap [\pi_i > \xi])\big]\nonumber \\
&\leq& \sum_{j \in I_0}\sum_{i \notin I_0}\left[P(\pi_i \leq \xi)  + P(\pi_j > \xi)\right]
 \leq  \sum_{j \in I_0}\sum_{i \notin I_0}\left[\xi + o(1)\right] \ \ \mbox{(by Lemma \ref{lemma:A1})} \nonumber \\
&=& d_0(d - d_0)\xi + o(1).  \nonumber
\eqa
Since $\xi$ is arbitrary, this shows that $\lim_{n \rightarrow \infty}P({\cal E}_n^c) = 0$, completing the proof.
\end{proof}

\begin{THEOREM} \label{thm.consistency}
With $\alpha_n$ the chosen bound of FDR (or Bonferroni), assume that $\alpha_n \to 0$ as $n \to \infty$ in such a way that
$\alpha_n = o(e^{-cn})$ for some constant c. 
Then, under assumptions of Lemma \ref{lemma:A1}, $$\lim_{n\to \infty}P(\hat{I} = I_0) = 1.$$
\end{THEOREM}

\begin{proof} 
Note that if the estimator $\hat{I}$ is equal to the set $I_0$, we have exactly $d_0$ rejections ($R = d_0$) with none of them being erroneous ($V = 0$). Therefore, consistency of $\hat{I}$ is verified by proving
\begin{equation}
P(\hat{I} = I_0) = P(R = d_0, V = 0) \rightarrow 1,\  \mbox{as}\  n \rightarrow \infty.
\end{equation}
This follows by showing that both $P(R \neq d_0)$ and $P(V \geq 1)$ are asymptotic negligible. By Lemma 2.1 in \cite{BuneaEtAl2006}, we have that 
\bqa
P(V \geq 1) \leq  P(R \neq d_0) + \frac{d_0(d-d_0)}{d}\alpha_n = o(1)
\eqa
by the choice of $\alpha_n \to 0$. Thus, in order to show consistency of $\hat{I}$ we only need to show that $P(R \neq d_0) \rightarrow 0$. Following \cite{BuneaEtAl2006}, consider the two cases: $\{R > d_0\}$ and $\{R < d_0\}$.

 For $\{R > d_0\}$, at least one false positive is made, then
    \begin{align}
        \exists j \not \in I_0 \text{ such that } \pi_{(k)} \leq \frac{k}{d}\frac{\alpha_n}{\sum_{l=1}^{d}l^{-1}}.
        \notag
    \end{align}
    This event can be written as
    \begin{align}
        \bigg\{\bigcup_{k=d_0+1}^{d}\bigg\{\pi_{(k)} \leq \frac{k}{d}\frac{\alpha_n}{\sum_{l=1}^{d}l^{-1}}\bigg\}\bigg\}.
        \notag
    \end{align}
    Each event in this union corresponds to each of the $d - d_0 - 1$ would be falsely rejected. 
    For $\{R < d_0\}$ there is at least one false negative.  For FDR, this can be expressed  as
    \begin{align}
         \bigg\{\pi_{(d_0)} \geq \frac{d_0}{d}\frac{\alpha_n}{\sum_{l=1}^{d}l^{-1}}\bigg\}
        \notag
    \end{align}

    Now the probability of the event $\{R \neq d_0\}$ can be bounded from above as follows
    \begin{align} \label{eq.R}
        P(\{R \neq d_0\}) &\leq P\bigg(\varepsilon_n\bigcap\bigg\{\pi_{(d_0)} \geq \frac{d_0}{d}\frac{\alpha_n}{\sum_{l=1}^{d}l^{-1}}\bigg\}\bigg)  \nonumber\\ 
        &\hspace{.5cm} + \sum_{k=d_0+1}^{d}P\bigg(\varepsilon_n\bigcap\bigg\{\pi_{(k)} \leq \frac{k}{d}\frac{\alpha_n}{\sum_{l=1}^{d}l^{-1}}\bigg\}\bigg) + P(\varepsilon_n^C)
    \end{align}
    Using Lemma \ref{lemma:epsilon}, it remains to show that the 2 first terms on the right hand side of (\ref{eq.R}) go to 0 as the sample size increases.  The first term on the right hand side of (\ref{eq.R}) can be bounded from above by
    \begin{align}
        P\left(\varepsilon_n\bigcap\left\{\pi_{(d_0)} \geq \frac{d_0}{d}\frac{\alpha_n}{\sum_{l=1}^{d}l^{-1}}\right\}\right) &\leq d_0\cdot \text{max}_{k\in I_0}P\left(\pi_k \geq \frac{d_0}{d}\frac{\alpha_n}{\sum_{l=1}^{d}l^{-1}}\right)
        \notag \\
        &= d_0\cdot\text{max}_{k\in I_0}P\left((1 - \Phi(|t_k|)) \geq \frac{d_0}{d}\frac{\alpha_n}{\sum_{l=1}^{d}l^{-1}}\right)
        \notag \\
        &= d_0\cdot\text{max}_{k\in I_0}P\left(|t_k| \leq \Phi^{-1}\left(1 - \frac{d_0}{d}\frac{\alpha_n}{\sum_{l=1}^{d}l^{-1}}\right)\right)
        \notag \\
        &= o(1) \text{ as } n \to \infty
        \notag
    \end{align}
    because for $k\in I_0$, $\gamma_k\neq 0$, and hence $t_j = O_p(\sqrt{n})$, $\frac{d_0}{d}\frac{\alpha_n}{\sum_{l=1}^{d}l^{-1}} \to 0$ and by the choice of $\alpha_n = o(exp\{-cn\})$. \\
    For the second term on the right hand side of (\ref{eq.R}) we have
    \begin{align}
    \sum_{k=d_0+1}^{d}P\bigg(\varepsilon_n\bigcap\bigg\{\pi_{(k)} \leq \frac{k}{d}\frac{\alpha_n}{\sum_{l=1}^{d}l^{-1}}\bigg\}\bigg) &\leq \sum_{k=d_0+1}^{d}P\bigg(\varepsilon_n\bigcap\bigg\{\pi_{(k)} \leq \frac{\alpha_n}{\sum_{l=1}^{d}l^{-1}}\bigg\}\bigg)
        \notag \\
        &\leq \sum_{k\not \in I_0}P\bigg(\pi_k \leq \frac{\alpha_n}{\sum_{l=1}^{d}l^{-1}}\bigg)
        \notag \\
        &= o(1) \text{ as } n \to \infty,
        \notag
    \end{align}
    since $\pi_k$ has a Uniform distribution under the null hypothesis $H_0^k \gamma_k = 0$ ($k\notin I_0$).
    
    The proof for Bonferroni follows similarly by noting that the cut-off is $\alpha_n/d$ instead of $\alpha_n/\sum_{l=1}^{d}l^{-1}$.

\end{proof}

\section{Simulations}

In this section we investigate the finite sample performance of the proposed method in recovering the true set of significant exogenous covariates through four simulation scenarios. First, in Scenario 1, we consider the GARCH(1,1)-X model
\begin{equation}
\left\{
\begin{aligned}
\varepsilon_t &= \sigma_tw_t \\
\sigma_t^2(\vtheta) &= \omega + \alpha\varepsilon_{t-1}^2 + \beta\sigma_{t-1}^2 + \gamma^T\vX_{t-1}
\end{aligned}
\right. \nonumber
\end{equation}
where $\vX_t = (X_{1,t}, X_{2,t}, \ldots, X_{d,t})$ is the vector of exogenous covariates. Here we set $\omega = 0.1$, $\alpha = 0.2$, $\beta = 0.4$, and consider $w_t$ to be Normal(0,1), Student $t$ with 5 and 7 degrees of freedom in order to evaluate cases with shocks that have heavier tails. We assess the proposed method's performance for $d = 5$ and $d = 8$ exogenous covariates $\vX_{t}$, which are generated as $X_{i, t} = e^{y_{i,t}}$, where $y_{i,t} = (0.2 +i0.01)y_{i,t-1} + e_t$, for $e_t \sim N(0,1)$. The parameter vector $\gamma$ is taken to be $\gamma = (0.25, 0, 0.3, 0.4, 0)$ for the case where $d = 5$ and $\gamma = (0.2, 0, 0.3, 0.3, 0, 0, 0, 0.8)$ for the case where $d = 8$. The choices of the parameter vector corresponding to the exogenous covariates enable us to effectively distinguish between truly significant predictors and those that do not influence the volatility model.

We run 500 Monte Carlo simulations for sample sizes $n = 500, 1000, 5000$, and $10000$, and compute the average number of exogenous covariates correctly selected and correctly excluded, as well as the average number of exogenous covariates incorrectly selected and incorrectly excluded. We also compute the proportion of time out of the 500 simulation runs that each exogenous covariate is selected. Finally, we assess the BIC and AIC of the selected model compared to the null model without any covariates and the full model with all covariates. The results are displayed in Tables \ref{tab.Indep_Corr_Sel} to \ref{tab.AIC}.


In the second simulation Scenario, denoted by Scenario 2, we consider the same data generating process as that of Scenario 1, except that the exogenous covariates are correlated. In this case, we use the absolute value of the random vector $\vX_t$, that is $|\vX_t|$, where $\vX_t$ is generated from a multivariate Normal distribution with mean $\bf{0}$ and covariance matrix $(\Sigma)_{ij} = \exp\{-0.5|i-j|\}$, that is, the covariance between exogenous covariates is an exponential decay. The results after 500 Monte Carlo simulations are displayed in Tables \ref{tab.corr_ExpDecay} to \ref{tab.AIC_ExpDecay}.

 The results suggest that the proposed method is  effective in identifying the truly relevant exogenous covariates across Normal and heavy tailed (Student t with 5 and 7 degrees of freedom) shocks for both 5 and 8 available covariates where 3 and 4 are relevant respectively. As the sample size increases, the method's ability to correctly select the relevant covariates improves significantly: the number of correctly selected and correctly excluded covariates increases with the increase in sample sizes in both the independent and exponential decay scenarios.

Moreover, Tables \ref{tab.Prop_time_d5}, \ref{tab.Prop_time_d8}, \ref{tab.covariates_d5_ExpDecay}, and \ref{tab.covariates_d8_ExpDecay} suggest that, as the sample size increases, the correct exogenous covariates (e.g., 1, 3, 4 for $d=5$; 1, 3, 4, 8 for $d=8$) are identified a high proportion of times in the simulations especially under the Normal shock distribution.
For instance, with a sample size of 5000 under Normal errors and $d=5$, covariates 1, 3, and 4 are selected 100\% of the time, while the irrelevant covariates (2 and 5) are selected in only 2-3\% of the simulations. Heavier-tailed shock distributions ($t_5$ and $t_7$) present a greater challenge, as evidenced by lower selection rates at smaller sample sizes, but the performance still converges towards perfect selection and exclusion as the sample size grows. The method also demonstrates a strong tendency to correctly exclude the irrelevant covariates, which are only selected a small percentage of time due to randomness.

Finally, the assessment of the proposed method using information criteria like BIC and AIC further supports the efficacy of the proposed procedure. In Tables \ref{tab.AIC} and \ref{tab.AIC_ExpDecay}, the average BIC and AIC values for the selected model are presented and compared to the null model (no covariates) and the full model (all covariates). The goal is for the selected model to achieve a favorable trade-off between model fit and complexity. The results suggest that the proposed method has larger AIC and BIC for small sample sizes, which is expected as it may more often miss relevant covariates compared to the full model. However, it
effectively navigates the trade-off with a smaller AIC and BIC for larger sample sizes, suggesting that the selected models are more parsimonious than the full model while capturing the necessary model complexity.

In addition to the Scenarios 1 and 2, a GARCH(2,1)-X model was also assessed. This case, with independent exogenous covariates, is denoted by Scenario 3, and the data is generated similarly to Scenario 1 but with parameters $\omega = 0.1$, $\valpha = (\alpha_1, \alpha_2) = (0.2, 0.15)$, $\beta = 0.4$, and $\gamma$ equal to the ones in Scenario 1. Finally, to assess the case of covariates with dependence, we consider Scenario 4, where we generate the data from the GARCH(2,1)-X model as in Scenario 3, but with covariates having exponentially decaying covariance as in Scenario 2. The results, presented in Tables \ref{tab:GARCH_Indep_correct} to \ref{tab:GARCH_ExpDecay_AIC}, suggest that, similarly to the GARCH(1,1)-X in Scenarios 1 and 2, as the sample size increases, the ability of the proposed method to select the truly relevant predictors and exclude irrelevant ones also improves for a GARCH(2,1)-X. Similarly, the proposed method achieves smaller BIC for larger sample sizes, while maintains an AIC similar to that of the full model, however with a more parsimonious model.

\begin{table}[ht]
\centering
\begin{tabular}{cccrrrr} 
\hline
\multirow{2}{*}{\textbf{d}} & \multirow{2}{*}{\textbf{Shock}} & \multirow{2}{*}{\textbf{Metric}} &
\multicolumn{4}{c}{\textbf{Sample size $n$}} \\
& & & \textbf{500} & \textbf{1000} & \textbf{5000} & \textbf{10000} \\
\hline
\multirow{12}{*}{d = 5}
  & \multirow{4}{*}{Normal}
    & corr. selected   & 1.96 & 2.82 & 3.00 & 3.00 \\
  & & incorr. selected & 0.04 & 0.05 & 0.05 & 0.06 \\
  & & corr. excluded   & 1.96 & 1.95 & 1.95 & 1.94 \\
  & & incorr. excluded & 1.04 & 0.18 & 0.00 & 0.00 \\
\cmidrule(lr){2-7}
& \multirow{4}{*}{t5}
    & corr. selected   & 0.54 & 1.47 & 2.94 & 2.99 \\
  & & incorr. selected & 0.01 & 0.02 & 0.03 & 0.02 \\
  & & corr. excluded   & 1.99 & 1.98 & 1.97 & 1.98 \\
  & & incorr. excluded & 2.46 & 1.53 & 0.06 & 0.01 \\
  \cmidrule(lr){2-7}
  & \multirow{4}{*}{t7}
    & corr. selected   & 0.87 & 2.04 & 3.00 & 3.00 \\
  & & incorr. selected & 0.03 & 0.02 & 0.05 & 0.04 \\
  & & corr. excluded   & 1.97 & 1.98 & 1.95 & 1.96 \\
  & & incorr. excluded & 2.13 & 0.96 & 0.00 & 0.00 \\
\hline
\multirow{12}{*}{d = 8}
  & \multirow{4}{*}{Normal}
    & corr. selected   & 1.20 & 2.27 & 3.99 & 4.00 \\
  & & incorr. selected & 0.02 & 0.03 & 0.07 & 0.07 \\
  & & corr. excluded   & 3.98 & 3.97 & 3.93 & 3.93 \\
  & & incorr. excluded & 2.80 & 1.73 & 0.01 & 0.00 \\
  \cmidrule(lr){2-7}
  & \multirow{4}{*}{t5}
    & corr. selected   & 0.27 & 0.82 & 3.09 & 3.74 \\
  & & incorr. selected & 0.01 & 0.00 & 0.02 & 0.02 \\
  & & corr. excluded   & 3.99 & 4.00 & 3.98 & 3.98 \\
  & & incorr. excluded & 3.73 & 3.18 & 0.91 & 0.26 \\
  \cmidrule(lr){2-7}
  & \multirow{4}{*}{t7}
    & corr. selected   & 0.51 & 1.31 & 3.68 & 3.98 \\
  & & incorr. selected & 0.01 & 0.02 & 0.05 & 0.05 \\
  & & corr. excluded   & 3.99 & 3.98 & 3.95 & 3.95 \\
  & & incorr. excluded & 3.49 & 2.69 & 0.32 & 0.02 \\
\hline
\end{tabular}
\caption{Average number of correctly and incorrectly selected and excluded exogenous covariates in the 500 simulation runs for Scenario 1 with independent covariates.} \label{tab.Indep_Corr_Sel}
\end{table}

\begin{table}[ht]
\centering
\begin{tabular}{crrrrrr}
\hline
\textbf{Shock} & \textbf{n} & $X_1$ & $X_2$ & $X_3$ & $X_4$ & $X_5$ \\
\hline
\multirow{4}{*}{Normal}
  & 500   & 0.53 & 0.02 & 0.62 & 0.81 & 0.02 \\
  & 1000  & 0.88 & 0.03 & 0.95 & 0.99 & 0.03 \\
  & 5000  & 1.00 & 0.02 & 1.00 & 1.00 & 0.03 \\
  & 10000 & 1.00 & 0.03 & 1.00 & 1.00 & 0.03 \\
\hline
\multirow{4}{*}{t5}
  & 500   & 0.11 & 0.01 & 0.17 & 0.26 & 0.01 \\
  & 1000  & 0.35 & 0.01 & 0.48 & 0.64 & 0.01 \\
  & 5000  & 0.97 & 0.02 & 0.98 & 1.00 & 0.01 \\
  & 10000 & 1.00 & 0.01 & 1.00 & 1.00 & 0.01 \\
\hline
\multirow{4}{*}{t7}
  & 500   & 0.20 & 0.01 & 0.26 & 0.41 & 0.02 \\
  & 1000  & 0.49 & 0.01 & 0.71 & 0.85 & 0.01 \\
  & 5000  & 1.00 & 0.03 & 1.00 & 1.00 & 0.01 \\
  & 10000 & 1.00 & 0.03 & 1.00 & 1.00 & 0.02 \\
\hline
\end{tabular}
\caption{Proportion of the 500 simulation runs that each exogenous covariate is selected by the proposed method for Scenario 1 where covariates are independent and $d=5$.}
\label{tab.Prop_time_d5}
\end{table}

\begin{table}[ht]
\centering
\begin{tabular}{crrrrrrrrr}
\hline
\textbf{Shock} & \textbf{n} & $X_1$ & $X_2$ & $X_3$ & $X_4$ & $X_5$ &$X_6$  & $X_7$ & $X_8$  \\
\hline
\multirow{4}{*}{Normal}
  & 500   & 0.04 & 0.00 & 0.18 & 0.14 & 0.00 & 0.00 & 0.01 & 0.84 \\
  & 1000  & 0.21 & 0.01 & 0.50 & 0.56 & 0.00 & 0.01 & 0.02 & 1.00 \\
  & 5000  & 0.99 & 0.02 & 1.00 & 1.00 & 0.02 & 0.01 & 0.02 & 1.00 \\
  & 10000 & 1.00 & 0.02 & 1.00 & 1.00 & 0.02 & 0.01 & 0.02 & 1.00 \\
\hline
\multirow{4}{*}{t5}
  & 500   & 0.02 & 0.00 & 0.02 & 0.02 & 0.00 & 0.00 & 0.01 & 0.21 \\
  & 1000  & 0.05 & 0.00 & 0.08 & 0.10 & 0.00 & 0.00 & 0.00 & 0.60 \\
  & 5000  & 0.46 & 0.01 & 0.84 & 0.79 & 0.00 & 0.00 & 0.01 & 1.00 \\
  & 10000 & 0.79 & 0.01 & 0.98 & 0.97 & 0.01 & 0.01 & 0.00 & 1.00 \\
\hline
\multirow{4}{*}{t7}
  & 500   & 0.01 & 0.00 & 0.04 & 0.05 & 0.00 & 0.00 & 0.00 & 0.41 \\
  & 1000  & 0.06 & 0.01 & 0.18 & 0.19 & 0.00 & 0.01 & 0.01 & 0.88 \\
  & 5000  & 0.74 & 0.01 & 0.95 & 0.99 & 0.02 & 0.02 & 0.01 & 1.00 \\
  & 10000 & 0.98 & 0.02 & 1.00 & 1.00 & 0.01 & 0.01 & 0.01 & 1.00 \\
\hline
\end{tabular}
\caption{Proportion of the 500 simulation runs that each exogenous covariate is selected by the proposed method for Scenario 1 where covariates are independent and $d=8$.}
\label{tab.Prop_time_d8}
\end{table}

\begin{table}[ht]
\centering
\resizebox{\textwidth}{!}{%
\begin{tabular}{ccrrrrrrr}
\hline
\multirow{2}{*}{\textbf{d}} & \multirow{2}{*}{\textbf{Shock}} & \multirow{2}{*}{\textbf{n}} &
\multicolumn{3}{c}{\textbf{BIC}} &
\multicolumn{3}{c}{\textbf{AIC}} \\
& & & \textbf{full} & \textbf{null} & \textbf{selected} & \textbf{full} & \textbf{null} & \textbf{selected} \\
\hline
\multirow{12}{*}{$d = 5$}
  & \multirow{4}{*}{Normal}
    & 500   & 2360.0 & 25690.0 & 4294.9 & 2326.3 & 25677.3 & 4273.8 \\
  & & 1000  & 4700.8 & 51584.9 & 4832.8 & 4661.5 & 51570.2 & 4804.0 \\
  & & 5000  & 23338.2 & 256862.0 & 23322.9 & 23286.1 & 256842.5 & 23283.5 \\
  & & 10000 & 46628.9 & 513827.7 & 46612.4 & 46571.2 & 513806.1 & 46568.7 \\
  \cmidrule(lr){2-9}
  & \multirow{4}{*}{t5}
    & 500   & 3064.9 & 41314.8 & 25538.6 & 3031.2 & 41302.2 & 25523.6 \\
  & & 1000  & 6124.8 & 82976.4 & 24989.5 & 6085.5 & 82961.7 & 24967.5 \\
  & & 5000  & 30652.0 & 418141.9 & 30644.6 & 30599.9 & 418122.3 & 30605.7 \\
  & & 10000 & 61349.6 & 837224.6 & 61340.0 & 61291.9 & 837203.0 & 61296.7 \\
  \cmidrule(lr){2-9}
  & \multirow{4}{*}{t7}
    & 500   & 2849.6 & 35882.7 & 17210.9 & 2815.9 & 35870.1 & 17194.5 \\
  & & 1000  & 5642.3 & 71082.4 & 8768.7 & 5603.1 & 71067.7 & 8743.9 \\
  & & 5000  & 28078.4 & 353161.1 & 28064.7 & 28026.2 & 353141.6 & 28025.3 \\
  & & 10000 & 56164.8 & 707240.6 & 56149.5 & 56107.1 & 707218.9 & 56105.9 \\
\hline
\multirow{12}{*}{$d = 8$}
  & \multirow{4}{*}{Normal}
    & 500   & 2879.8 & 34682.2 & 7202.1 & 2833.5 & 34669.6 & 7184.4 \\
  & & 1000  & 5717.0 & 69151.1 & 5694.2 & 5663.0 & 69136.4 & 5668.1 \\
  & & 5000  & 28375.6 & 347174.1 & 28344.7 & 28303.9 & 347154.6 & 28298.7 \\
  & & 10000 & 56691.3 & 692769.3 & 56657.7 & 56612.0 & 692747.7 & 56606.7 \\
  \cmidrule(lr){2-9}
  & \multirow{4}{*}{t5}
    & 500   & 3596.2 & 55258.5 & 42988.1 & 3549.8 & 55245.9 & 42974.3 \\
  & & 1000  & 7160.6 & 109869.1 & 45238.7 & 7106.6 & 109854.4 & 45219.9 \\
  & & 5000  & 35702.1 & 547147.6 & 35697.9 & 35630.4 & 547128.0 & 35658.1 \\
  & & 10000 & 71268.2 & 1088040.5 & 71248.9 & 71188.9 & 1088018.9 & 71200.2 \\
  \cmidrule(lr){2-9}
  & \multirow{4}{*}{t7}
    & 500   & 3352.3 & 46767.7 & 27187.9 & 3305.9 & 46755.0 & 27173.0 \\
  & & 1000  & 6673.2 & 93441.8 & 15427.5 & 6619.2 & 93427.1 & 15406.2 \\
  & & 5000  & 33079.9 & 464530.2 & 33053.6 & 33008.2 & 464510.6 & 33009.7 \\
  & & 10000 & 66213.0 & 930729.7 & 66182.4 & 66133.7 & 930708.1 & 66131.8 \\
\hline
\end{tabular}
}
\caption{Average BIC and AIC in the 500 simulation runs for Scenario 1 with independent exogenous covariates.}
\label{tab.AIC}
\end{table}

\begin{table}[ht]
\centering
\begin{tabular}{cccrrrr}
\hline
\multirow{2}{*}{\textbf{d}} & \multirow{2}{*}{\textbf{Shock}} & \multirow{2}{*}{\textbf{Metric}} &
\multicolumn{4}{c}{\textbf{Sample size $n$}} \\
& & & \textbf{500} & \textbf{1000} & \textbf{5000} & \textbf{10000} \\
\hline
\multirow{12}{*}{$d = 5$}
  & \multirow{4}{*}{Normal}
    & corr. selected   & 0.63 & 1.37 & 2.97 & 3.00 \\
  & & incorr. selected & 0.05 & 0.04 & 0.07 & 0.06 \\
  & & corr. excluded   & 1.95 & 1.96 & 1.93 & 1.94 \\
  & & incorr. excluded & 2.37 & 1.63 & 0.03 & 0.00 \\
\cmidrule(lr){2-7}
  & \multirow{4}{*}{t5}
    & corr. selected   & 0.19 & 0.38 & 1.82 & 2.64 \\
  & & incorr. selected & 0.03 & 0.04 & 0.04 & 0.07 \\
  & & corr. excluded   & 1.97 & 1.96 & 1.96 & 1.93 \\
  & & incorr. excluded & 2.81 & 2.62 & 1.18 & 0.36 \\
\cmidrule(lr){2-7}
  & \multirow{4}{*}{t7}
    & corr. selected   & 0.30 & 0.56 & 2.54 & 2.96 \\
  & & incorr. selected & 0.02 & 0.03 & 0.06 & 0.07 \\
  & & corr. excluded   & 1.98 & 1.97 & 1.94 & 1.93 \\
  & & incorr. excluded & 2.70 & 2.44 & 0.46 & 0.04 \\
\hline
\multirow{12}{*}{$d = 8$}
  & \multirow{4}{*}{Normal}
    & corr. selected   & 0.38 & 0.65 & 3.61 & 3.98 \\
  & & incorr. selected & 0.04 & 0.03 & 0.05 & 0.07 \\
  & & corr. excluded   & 3.96 & 3.97 & 3.95 & 3.93 \\
  & & incorr. excluded & 3.62 & 3.35 & 0.39 & 0.02 \\
\cmidrule(lr){2-7}
  & \multirow{4}{*}{t5}
    & corr. selected   & 0.09 & 0.13 & 1.22 & 2.62 \\
  & & incorr. selected & 0.01 & 0.01 & 0.04 & 0.04 \\
  & & corr. excluded   & 3.99 & 3.99 & 3.96 & 3.96 \\
  & & incorr. excluded & 3.91 & 3.87 & 2.78 & 1.38 \\
\cmidrule(lr){2-7}
  & \multirow{4}{*}{t7}
    & corr. selected   & 0.13 & 0.25 & 2.25 & 3.46 \\
  & & incorr. selected & 0.02 & 0.03 & 0.03 & 0.06 \\
  & & corr. excluded   & 3.98 & 3.97 & 3.97 & 3.94 \\
  & & incorr. excluded & 3.87 & 3.75 & 1.75 & 0.54 \\
\hline
\end{tabular}
\caption{Average number of correctly and incorrectly selected and excluded exogenous covariates in the 500 simulation runs for Scenario 2 where covariates have exponentially decaying covariance.}
\label{tab.corr_ExpDecay}
\end{table}

\begin{table}[ht]
\centering
\begin{tabular}{crrrrrr}
\hline
\textbf{Shock} & \textbf{n} & $X_1$ & $X_2$ & $X_3$ & $X_4$ & $X_5$  \\
\hline
\multirow{4}{*}{Normal}
  & 500   & 0.30 & 0.03 & 0.16 & 0.17 & 0.03 \\
  & 1000  & 0.63 & 0.02 & 0.39 & 0.35 & 0.02 \\
  & 5000  & 1.00 & 0.04 & 0.99 & 0.98 & 0.02 \\
  & 10000 & 1.00 & 0.02 & 1.00 & 1.00 & 0.04 \\
\hline
\multirow{4}{*}{t5}
  & 500   & 0.07 & 0.02 & 0.07 & 0.05 & 0.01 \\
  & 1000  & 0.17 & 0.02 & 0.13 & 0.08 & 0.02 \\
  & 5000  & 0.80 & 0.03 & 0.54 & 0.48 & 0.01 \\
  & 10000 & 0.97 & 0.02 & 0.83 & 0.84 & 0.05 \\
\hline
\multirow{4}{*}{t7}
  & 500   & 0.15 & 0.01 & 0.06 & 0.09 & 0.01 \\
  & 1000  & 0.24 & 0.01 & 0.15 & 0.17 & 0.02 \\
  & 5000  & 0.97 & 0.04 & 0.77 & 0.81 & 0.02 \\
  & 10000 & 1.00 & 0.03 & 0.97 & 0.98 & 0.04 \\
\hline
\end{tabular}
\caption{Proportion of the 500 simulation runs that each exogenous covariate is selected by the proposed method for Scenario 2 where covariates have exponentially decaying covariance and $d=5$.}
\label{tab.covariates_d5_ExpDecay}
\end{table}

\begin{table}[ht]
\centering
\begin{tabular}{crrrrrrrrr}
\hline
\textbf{Shock} & \textbf{n} & $X_1$ & $X_2$ & $X_3$ & $X_4$ & $X_5$ &$X_6$  & $X_7$ & $X_8$ \\
\hline
\multirow{4}{*}{Normal}
  & 500   & 0.18 & 0.01 & 0.08 & 0.08 & 0.01 & 0.01 & 0.00 & 0.05 \\
  & 1000  & 0.29 & 0.01 & 0.16 & 0.09 & 0.00 & 0.01 & 0.01 & 0.11 \\
  & 5000  & 1.00 & 0.02 & 0.96 & 0.77 & 0.01 & 0.01 & 0.01 & 0.88 \\
  & 10000 & 1.00 & 0.01 & 1.00 & 0.99 & 0.02 & 0.02 & 0.02 & 0.99 \\
\hline
\multirow{4}{*}{t5}
  & 500   & 0.03 & 0.01 & 0.02 & 0.01 & 0.00 & 0.00 & 0.00 & 0.02 \\
  & 1000  & 0.05 & 0.00 & 0.04 & 0.02 & 0.01 & 0.00 & 0.00 & 0.02 \\
  & 5000  & 0.46 & 0.02 & 0.32 & 0.20 & 0.01 & 0.00 & 0.01 & 0.23 \\
  & 10000 & 0.83 & 0.01 & 0.74 & 0.48 & 0.01 & 0.01 & 0.02 & 0.57 \\
\hline
\multirow{4}{*}{t7}
  & 500   & 0.06 & 0.01 & 0.03 & 0.02 & 0.00 & 0.00 & 0.00 & 0.02 \\
  & 1000  & 0.10 & 0.01 & 0.07 & 0.04 & 0.01 & 0.01 & 0.01 & 0.03 \\
  & 5000  & 0.77 & 0.02 & 0.61 & 0.40 & 0.01 & 0.00 & 0.01 & 0.48 \\
  & 10000 & 0.98 & 0.03 & 0.95 & 0.76 & 0.01 & 0.01 & 0.01 & 0.77 \\
\hline
\end{tabular}
\caption{Proportion of the 500 simulation runs that each exogenous covariate is selected by the proposed method for Scenario 2 where covariates have exponentially decaying covariance and $d=8$.}
\label{tab.covariates_d8_ExpDecay}
\end{table}

\begin{table}[ht]
\centering
\resizebox{\textwidth}{!}{%
\begin{tabular}{ccrrrrrrr}
\hline
\multirow{2}{*}{\textbf{d}} & \multirow{2}{*}{\textbf{Shock}} & \multirow{2}{*}{\textbf{n}} &
\multicolumn{3}{c}{\textbf{BIC}} &
\multicolumn{3}{c}{\textbf{AIC}} \\
& & & \textbf{full} & \textbf{null} & \textbf{selected} & \textbf{full} & \textbf{null} & \textbf{selected} \\
\hline
\multirow{12}{*}{$d=5$}
  & \multirow{4}{*}{Normal}
    & 500   & 1356.7 & 12155.6 & 6445.7 & 1323.0 & 12143.0 & 6430.2 \\
  & & 1000  & 2696.5 & 24402.8 & 6653.6 & 2657.2 & 24388.1 & 6632.0 \\
  & & 5000  & 13312.7 & 121458.0 & 13297.4 & 13260.6 & 121438.5 & 13258.1 \\
  & & 10000 & 26580.6 & 242932.1 & 26564.0 & 26522.9 & 242910.4 & 26520.3 \\
  \cmidrule(lr){2-9}
  & \multirow{4}{*}{t5}
    & 500   & 2059.1 & 21616.5 & 17800.0 & 2025.4 & 21603.9 & 17766.5 \\
  & & 1000  & 4112.9 & 43432.8 & 30209.3 & 4073.6 & 43418.1 & 30192.5 \\
  & & 5000  & 20573.1 & 218766.5 & 39493.2 & 20521.0 & 218747.0 & 39461.5 \\
  & & 10000 & 41210.6 & 439305.9 & 45579.5 & 41152.9 & 439284.3 & 45538.3 \\
  \cmidrule(lr){2-9}
  & \multirow{4}{*}{t7}
    & 500   & 1841.2 & 18249.6 & 13720.0 & 1807.4 & 18236.9 & 13706.0 \\
  & & 1000  & 3634.0 & 36050.6 & 21517.1 & 3594.7 & 36035.9 & 21499.5 \\
  & & 5000  & 18020.8 & 179290.2 & 20112.0 & 17968.6 & 179270.7 & 20075.5 \\
  & & 10000 & 36046.0 & 358855.6 & 36030.6 & 35988.3 & 358834.0 & 35987.2 \\
\hline
\multirow{12}{*}{$d=8$}
  & \multirow{4}{*}{Normal}
    & 500   & 1435.3 & 12731.8 & 8807.6 & 1388.9 & 12719.1 & 8793.2 \\
  & & 1000  & 2830.6 & 25445.6 & 14686.4 & 2776.6 & 25430.9 & 14668.3 \\
  & & 5000  & 13937.8 & 127172.8 & 13905.9 & 13866.2 & 127153.2 & 13862.5 \\
  & & 10000 & 27819.0 & 254329.8 & 27785.7 & 27739.7 & 254308.1 & 27734.8 \\
  \cmidrule(lr){2-9}
  & \multirow{4}{*}{t5}
    & 500   & 2146.9 & 22949.2 & 21332.5 & 2100.5 & 22936.6 & 21319.4 \\
  & & 1000  & 4267.8 & 45813.7 & 40739.6 & 4213.8 & 45799.0 & 40724.2 \\
  & & 5000  & 21200.0 & 228062.5 & 76013.4 & 21128.3 & 228042.9 & 75985.6 \\
  & & 10000 & 42276.8 & 453163.9 & 56014.1 & 42197.5 & 453142.3 & 55973.2 \\
  \cmidrule(lr){2-9}
  & \multirow{4}{*}{t7}
    & 500   & 1907.9 & 18828.2 & 16556.4 & 1861.5 & 18815.6 & 16543.1 \\
  & & 1000  & 3779.3 & 37472.0 & 29650.0 & 3725.3 & 37457.3 & 29633.9 \\
  & & 5000  & 18612.5 & 186089.4 & 27452.4 & 18540.8 & 186069.9 & 27418.0 \\
  & & 10000 & 37267.7 & 373475.3 & 37239.5 & 37188.4 & 373453.7 & 37192.5 \\
\hline
\end{tabular}
}
\caption{Average BIC and AIC of the 500 simulation runs for Scenario 2 with  exogenous covariates with exponentially decaying covariance.}
\label{tab.AIC_ExpDecay}
\end{table}

\begin{table}[ht]
\centering
\begin{tabular}{cccrrrr}
\hline
\multirow{2}{*}{\textbf{$d$}} & \multirow{2}{*}{\textbf{Shock}} & \multirow{2}{*}{\textbf{Metric}} &
\multicolumn{4}{c}{\textbf{Sample size $n$}} \\
& & & \textbf{500} & \textbf{1000} & \textbf{5000} & \textbf{10000} \\
\hline
\multirow{12}{*}{$d = 5$}
  & \multirow{4}{*}{Normal}
    & corr. select   & 1.05 & 2.36 & 3.00 & 3.00 \\
  & & incorr. select & 0.03 & 0.03 & 0.04 & 0.07 \\
  & & corr. exclude  & 1.97 & 1.97 & 1.96 & 1.93 \\
  & & incorr. exclude & 1.95 & 0.64 & 0.00 & 0.00 \\
  \cmidrule(lr){2-7}
  & \multirow{4}{*}{t5}
    & corr. select   & 0.14 & 0.50 & 2.57 & 2.95 \\
  & & incorr. select & 0.01 & 0.00 & 0.03 & 0.02 \\
  & & corr. exclude  & 1.99 & 2.00 & 1.97 & 1.98 \\
  & & incorr. exclude & 2.86 & 2.50 & 0.43 & 0.05 \\
  \cmidrule(lr){2-7}
  & \multirow{4}{*}{t7}
    & corr. select   & 0.23 & 0.93 & 2.96 & 3.00 \\
  & & incorr. select & 0.01 & 0.01 & 0.04 & 0.05 \\
  & & corr. exclude  & 1.99 & 1.99 & 1.96 & 1.95 \\
  & & incorr. exclude & 2.77 & 2.07 & 0.04 & 0.00 \\
\hline 
\multirow{12}{*}{$d = 8$}
  & \multirow{4}{*}{Normal}
    & corr. select   & 0.65 & 1.65 & 3.84 & 4.00 \\
  & & incorr. select & 0.01 & 0.02 & 0.08 & 0.08 \\
  & & corr. exclude  & 3.99 & 3.98 & 3.92 & 3.92 \\
  & & incorr. exclude & 3.35 & 2.35 & 0.16 & 0.00 \\
  \cmidrule(lr){2-7}
  & \multirow{4}{*}{t5}
    & corr. select   & 0.06 & 0.26 & 1.80 & 2.93 \\
  & & incorr. select & 0.00 & 0.00 & 0.01 & 0.01 \\
  & & corr. exclude  & 4.00 & 4.00 & 3.99 & 3.99 \\
  & & incorr. exclude & 3.94 & 3.74 & 2.20 & 1.07 \\
  \cmidrule(lr){2-7}
  & \multirow{4}{*}{t7}
    & corr. select   & 0.16 & 0.60 & 2.73 & 3.65 \\
  & & incorr. select & 0.00 & 0.00 & 0.04 & 0.04 \\
  & & corr. exclude  & 4.00 & 4.00 & 3.96 & 3.96 \\
  & & incorr. exclude & 3.84 & 3.40 & 1.27 & 0.35 \\
\hline
\end{tabular}
\caption{Average number of correctly and incorrectly selected and excluded exogenous covariates in the 500 simulation runs for Scenario 3 with independent covariates in a GARCH(2,1)-X model.} \label{tab:GARCH_Indep_correct}
\end{table}

\begin{table}[ht]
\centering
\begin{tabular}{crrrrrr}
\hline
\textbf{Shock} & \textbf{n} & $X_1$ & $X_2$ & $X_3$ & $X_4$ & $X_5$  \\
\hline
\multirow{4}{*}{Normal}
  & 500   & 0.22 & 0.02 & 0.34 & 0.49 & 0.02 \\
  & 1000  & 0.65 & 0.02 & 0.77 & 0.95 & 0.01 \\
  & 5000  & 1.00 & 0.02 & 1.00 & 1.00 & 0.02 \\
  & 10000 & 1.00 & 0.04 & 1.00 & 1.00 & 0.03 \\
\hline
\multirow{4}{*}{t5}
  & 500   & 0.03 & 0.00 & 0.03 & 0.08 & 0.01 \\
  & 1000  & 0.10 & 0.00 & 0.15 & 0.24 & 0.00 \\
  & 5000  & 0.77 & 0.01 & 0.85 & 0.95 & 0.01 \\
  & 10000 & 0.97 & 0.01 & 0.99 & 0.99 & 0.01 \\
\hline
\multirow{4}{*}{t7}
  & 500   & 0.05 & 0.00 & 0.07 & 0.11 & 0.01 \\
  & 1000  & 0.19 & 0.00 & 0.29 & 0.45 & 0.00 \\
  & 5000  & 0.96 & 0.03 & 1.00 & 1.00 & 0.01 \\
  & 10000 & 1.00 & 0.03 & 1.00 & 1.00 & 0.02 \\
\hline
\end{tabular}
\caption{Proportion of the 500 simulation runs that each exogenous covariate is selected by the proposed method for Scenario 3 from a GARCH(2,1)-X with independent covariates and $d=5$.}
\label{tab.GARCH_Indep_d5}
\end{table}

\begin{table}[ht]
\centering
\begin{tabular}{crrrrrrrrr}
\hline
\textbf{Shock} & \textbf{n} & $X_1$ & $X_2$ & $X_3$ & $X_4$ & $X_5$ &$X_6$  & $X_7$ & $X_8$ \\
\hline
\multirow{4}{*}{Normal}
  & 500   & 0.02 & 0.00 & 0.07 & 0.05 & 0.01 & 0.00 & 0.00 & 0.51 \\
  & 1000  & 0.12 & 0.00 & 0.28 & 0.28 & 0.00 & 0.00 & 0.01 & 0.97 \\
  & 5000  & 0.87 & 0.02 & 0.99 & 0.98 & 0.02 & 0.01 & 0.02 & 1.00 \\
  & 10000 & 1.00 & 0.02 & 1.00 & 1.00 & 0.02 & 0.01 & 0.03 & 1.00 \\
\hline
\multirow{4}{*}{t5}
  & 500   & 0.00 & 0.00 & 0.01 & 0.01 & 0.00 & 0.00 & 0.00 & 0.05 \\
  & 1000  & 0.01 & 0.00 & 0.03 & 0.03 & 0.00 & 0.00 & 0.00 & 0.19 \\
  & 5000  & 0.12 & 0.00 & 0.38 & 0.33 & 0.00 & 0.00 & 0.01 & 0.97 \\
  & 10000 & 0.41 & 0.00 & 0.75 & 0.77 & 0.00 & 0.00 & 0.00 & 1.00 \\
\hline
\multirow{4}{*}{t7}
  & 500   & 0.00 & 0.00 & 0.01 & 0.00 & 0.00 & 0.00 & 0.00 & 0.14 \\
  & 1000  & 0.02 & 0.00 & 0.03 & 0.06 & 0.00 & 0.00 & 0.00 & 0.49 \\
  & 5000  & 0.34 & 0.01 & 0.68 & 0.70 & 0.01 & 0.01 & 0.01 & 1.00 \\
  & 10000 & 0.71 & 0.02 & 0.98 & 0.97 & 0.01 & 0.01 & 0.01 & 1.00 \\
\hline
\end{tabular}
\caption{Proportion of the 500 simulation runs that each exogenous covariate is selected by the proposed method for Scenario 3 from a GARCH(2,1)-X with independent covariates and $d=8$.}
\label{tab.GARCH_Indep_d8}
\end{table}

\begin{table}[http!]
\centering
\resizebox{\textwidth}{!}{%
\begin{tabular}{ccrrrrrrr}
\hline
\multirow{2}{*}{\textbf{d}} & \multirow{2}{*}{\textbf{Shock}} & \multirow{2}{*}{\textbf{n}} &
\multicolumn{3}{c}{\textbf{BIC}} &
\multicolumn{3}{c}{\textbf{AIC}} \\
& & & \textbf{full} & \textbf{null} & \textbf{selected} & \textbf{full} & \textbf{null} & \textbf{selected} \\
\hline
\multirow{12}{*}{$d = 5$}
  & \multirow{4}{*}{Normal}
    & 500   & 2747.7 & 14940.4 & 6921.7 & 2709.8 & 14923.6 & 6900.3 \\
  & & 1000  & 5481.2 & 29944.1 & 5704.7 & 5437.0 & 29924.5 & 5673.3 \\
  & & 5000  & 27210.4 & 149837.0 & 27195.1 & 27151.8 & 149810.9 & 27149.2 \\
  & & 10000 & 54370.8 & 299221.1 & 54354.2 & 54305.9 & 299192.3 & 54303.2 \\
  \cmidrule(lr){2-9}
  & \multirow{4}{*}{t5}
    & 500   & 3799.7 & 22674.7 & 19894.6 & 3761.8 & 22657.9 & 19877.1 \\
  & & 1000  & 7612.4 & 45175.8 & 31586.1 & 7568.2 & 45156.1 & 31564.0 \\
  & & 5000  & 38054.2 & 227819.8 & 41572.8 & 37995.5 & 227793.7 & 41529.8 \\
  & & 10000 & 76215.7 & 455202.2 & 76206.6 & 76150.8 & 455173.3 & 76156.4 \\
  \cmidrule(lr){2-9}
  & \multirow{4}{*}{t7}
    & 500   & 3452.4 & 19969.9 & 16710.5 & 3414.4 & 19953.0 & 16692.6 \\
  & & 1000  & 6833.8 & 39475.7 & 18591.4 & 6789.6 & 39456.1 & 18567.1 \\
  & & 5000  & 34010.6 & 196706.4 & 34546.4 & 33951.9 & 196680.4 & 34500.8 \\
  & & 10000 & 68059.9 & 393894.7 & 68044.7 & 67995.0 & 393865.9 & 67993.8 \\
\hline 
\multirow{12}{*}{$d = 8$}
  & \multirow{4}{*}{Normal}
    & 500   & 3269.8 & 17896.9 & 9683.6 & 3219.2 & 17879.9 & 9664.0 \\
  & & 1000  & 6494.6 & 35613.2 & 7247.2 & 6435.7 & 35593.5 & 7219.4 \\
  & & 5000  & 32248.8 & 179298.3 & 32217.6 & 32170.6 & 179272.3 & 32166.0 \\
  & & 10000 & 64431.8 & 357412.9 & 64398.1 & 64345.2 & 357384.1 & 64339.9 \\
  \cmidrule(lr){2-9}
  & \multirow{4}{*}{t5}
    & 500   & 4337.9 & 26421.7 & 25137.2 & 4287.4 & 26404.8 & 25120.1 \\
  & & 1000  & 8648.8 & 52307.1 & 42208.9 & 8589.9 & 52287.5 & 42188.0 \\
  & & 5000  & 43097.9 & 261417.3 & 50118.4 & 43019.8 & 261391.2 & 50080.5 \\
  & & 10000 & 86046.8 & 520224.7 & 86048.6 & 85960.3 & 520195.9 & 85998.5 \\
  \cmidrule(lr){2-9}
  & \multirow{4}{*}{t7}
    & 500   & 3955.0 & 22770.8 & 19896.2 & 3904.4 & 22753.9 & 19878.7 \\
  & & 1000  & 7877.3 & 46011.9 & 26219.9 & 7818.4 & 45992.2 & 26197.3 \\
  & & 5000  & 39001.3 & 228140.3 & 38980.1 & 38923.1 & 228114.3 & 38935.9 \\
  & & 10000 & 78096.1 & 456673.2 & 78068.7 & 78009.6 & 456644.3 & 78013.2 \\
\hline
\end{tabular}
}
\caption{Average BIC and AIC of the 500 simulation runs for Scenario 3 with independent exogenous covariates in a GARCH(2,1)-X model.} \label{tab:GARCH_Indep_AIC}
\end{table}

\begin{table}[ht]
\centering
\begin{tabular}{cccrrrr}
\hline
\multirow{2}{*}{\textbf{$d$}} & \multirow{2}{*}{\textbf{Shock}} & \multirow{2}{*}{\textbf{Metric}} &
\multicolumn{4}{c}{\textbf{Sample size $n$}} \\
& & & \textbf{500} & \textbf{1000} & \textbf{5000} & \textbf{10000} \\
\hline
\multirow{12}{*}{$d = 5$}
  & \multirow{4}{*}{Normal}
    & corr. select   & 0.35 & 0.75 & 2.73 & 3.00 \\
  & & incorr. select & 0.04 & 0.03 & 0.05 & 0.06 \\
  & & corr. exclude  & 1.96 & 1.97 & 1.95 & 1.94 \\
  & & incorr. exclude & 2.65 & 2.25 & 0.27 & 0.00 \\
  \cmidrule(lr){2-7}
  & \multirow{4}{*}{t5}
    & corr. select   & 0.08 & 0.16 & 0.87 & 1.66 \\
  & & incorr. select & 0.01 & 0.03 & 0.02 & 0.03 \\
  & & corr. exclude  & 1.99 & 1.97 & 1.98 & 1.97 \\
  & & incorr. exclude & 2.92 & 2.84 & 2.13 & 1.34 \\
  \cmidrule(lr){2-7}
  & \multirow{4}{*}{t7}
    & corr. select   & 0.14 & 0.21 & 1.64 & 2.50 \\
  & & incorr. select & 0.01 & 0.02 & 0.04 & 0.05 \\
  & & corr. exclude  & 1.99 & 1.98 & 1.96 & 1.95 \\
  & & incorr. exclude & 2.86 & 2.79 & 1.36 & 0.50 \\
\hline 
\multirow{12}{*}{$d = 8$}
  & \multirow{4}{*}{Normal}
    & corr. select   & 0.18 & 0.31 & 2.69 & 3.74 \\
  & & incorr. select & 0.02 & 0.02 & 0.03 & 0.09 \\
  & & corr. exclude  & 3.98 & 3.98 & 3.97 & 3.91 \\
  & & incorr. exclude & 3.82 & 3.69 & 1.31 & 0.26 \\
  \cmidrule(lr){2-7}
  & \multirow{4}{*}{t5}
    & corr. select   & 0.03 & 0.05 & 0.35 & 1.00 \\
  & & incorr. select & 0.01 & 0.01 & 0.02 & 0.01 \\
  & & corr. exclude  & 3.99 & 3.99 & 3.98 & 3.99 \\
  & & incorr. exclude & 3.97 & 3.95 & 3.65 & 3.00 \\
  \cmidrule(lr){2-7}
  & \multirow{4}{*}{t7}
    & corr. select   & 0.04 & 0.09 & 1.02 & 2.11 \\
  & & incorr. select & 0.02 & 0.02 & 0.01 & 0.03 \\
  & & corr. exclude  & 3.98 & 3.98 & 3.99 & 3.97 \\
  & & incorr. exclude & 3.96 & 3.91 & 2.98 & 1.89 \\
\hline
\end{tabular}
\caption{Average number of correctly and incorrectly selected and excluded exogenous covariates in the 500 simulation runs for Scenario 4 with covariates having exponential decaying covariance in a GARCH(2,1)-X model.} \label{tab:GARCH_ExpDecay_correct}
\end{table}

\begin{table}[ht]
\centering
\begin{tabular}{crrrrrr}
\hline
\textbf{Shock} & \textbf{n} & $X_1$ & $X_2$ & $X_3$ & $X_4$ & $X_5$ \\
\hline
\multirow{4}{*}{Normal}
 & 500 & 0.16 & 0.02 & 0.08 & 0.11 & 0.02 \\
 & 1000 & 0.38 & 0.01 & 0.21 & 0.16 & 0.02 \\
 & 5000 & 0.99 & 0.03 & 0.87 & 0.88 & 0.02 \\
 & 10000 & 1.00 & 0.02 & 1.00 & 1.00 & 0.04 \\
\hline
\multirow{4}{*}{t5}
 & 500 & 0.04 & 0.01 & 0.02 & 0.02 & 0.00 \\
 & 1000 & 0.06 & 0.02 & 0.06 & 0.03 & 0.01 \\
 & 5000 & 0.43 & 0.01 & 0.24 & 0.20 & 0.01 \\
 & 10000 & 0.74 & 0.01 & 0.45 & 0.48 & 0.02 \\
\hline
\multirow{4}{*}{t7}
 & 500 & 0.05 & 0.01 & 0.04 & 0.05 & 0.01 \\
 & 1000 & 0.09 & 0.00 & 0.06 & 0.05 & 0.01 \\
 & 5000 & 0.72 & 0.03 & 0.45 & 0.47 & 0.02 \\
 & 10000 & 0.97 & 0.02 & 0.76 & 0.77 & 0.03 \\
\hline
\end{tabular}
\caption{Proportion of the 500 simulation runs that each exogenous covariate is selected by the proposed method for Scenario 4 from a GARCH(2,1)-X for covariates with exponentially decaying covariance and $d=5$.}
\label{tab:GARCH_ExpDecay_d5}
\end{table}

\begin{table}[ht]
\centering
\begin{tabular}{crrrrrrrrr}
\hline
\textbf{Shock} & \textbf{n} & $X_1$ & $X_2$ & $X_3$ & $X_4$ & $X_5$ &$X_6$  & $X_7$ & $X_8$ \\
\hline
\multirow{4}{*}{Normal}
 & 500 & 0.07 & 0.01 & 0.04 & 0.04 & 0.01 & 0.00 & 0.00 & 0.03 \\
 & 1000 & 0.15 & 0.01 & 0.06 & 0.05 & 0.00 & 0.00 & 0.01 & 0.05 \\
 & 5000 & 0.87 & 0.01 & 0.77 & 0.50 & 0.01 & 0.01 & 0.01 & 0.55 \\
 & 10000 & 1.00 & 0.02 & 0.97 & 0.88 & 0.02 & 0.02 & 0.03 & 0.89 \\
\hline
\multirow{4}{*}{t5}
 & 500 & 0.01 & 0.01 & 0.01 & 0.00 & 0.00 & 0.00 & 0.00 & 0.00 \\
 & 1000 & 0.02 & 0.01 & 0.01 & 0.01 & 0.00 & 0.00 & 0.00 & 0.01 \\
 & 5000 & 0.15 & 0.01 & 0.10 & 0.05 & 0.00 & 0.00 & 0.00 & 0.05 \\
 & 10000 & 0.39 & 0.00 & 0.28 & 0.16 & 0.00 & 0.00 & 0.00 & 0.17 \\
\hline
\multirow{4}{*}{t7}
 & 500 & 0.01 & 0.01 & 0.01 & 0.01 & 0.00 & 0.00 & 0.01 & 0.01 \\
 & 1000 & 0.03 & 0.01 & 0.03 & 0.01 & 0.00 & 0.00 & 0.00 & 0.01 \\
 & 5000 & 0.38 & 0.01 & 0.28 & 0.19 & 0.00 & 0.00 & 0.00 & 0.16 \\
 & 10000 & 0.72 & 0.00 & 0.62 & 0.37 & 0.01 & 0.00 & 0.01 & 0.39 \\
\hline
\end{tabular}
\caption{Proportion of the 500 simulation runs that each exogenous covariate is selected by the proposed method for Scenario 4 from a GARCH(2,1)-X for covariates with exponentially decaying covariance and $d=8$.}
\label{tab:GARCH_ExpDecay_d8}
\end{table}

\begin{table}[ht]
\centering
\resizebox{\textwidth}{!}{%
\begin{tabular}{ccrrrrrrr}
\hline
\multirow{2}{*}{\textbf{d}} & \multirow{2}{*}{\textbf{Shock}} & \multirow{2}{*}{\textbf{n}} &
\multicolumn{3}{c}{\textbf{BIC}} &
\multicolumn{3}{c}{\textbf{AIC}} \\
& & & \textbf{full} & \textbf{null} & \textbf{selected} & \textbf{full} & \textbf{null} & \textbf{selected} \\
\hline
\multirow{12}{*}{$d = 5$}
 & \multirow{4}{*}{Normal}
   & 500   & 1731.8 & 9116.6 & 6716.1 & 1693.9 & 9099.7 & 6697.6 \\
 & & 1000  & 3449.3 & 18267.4 & 9845.5 & 3405.1 & 18247.8 & 9822.0 \\
 & & 5000  & 17052.2 & 91232.4 & 17278.6 & 16993.6 & 91206.3 & 17234.4 \\
 & & 10000 & 34054.5 & 182303.9 & 34037.8 & 33989.6 & 182275.1 & 33986.9 \\
  \cmidrule(lr){2-9}
 & \multirow{4}{*}{t5}
   & 500   & 2779.8 & 15761.1 & 14670.7 & 2741.9 & 15744.3 & 14653.5 \\
 & & 1000  & 5568.9 & 31488.0 & 26847.5 & 5524.8 & 31468.4 & 26827.0 \\
 & & 5000  & 27820.8 & 158730.9 & 76978.0 & 27762.1 & 158704.9 & 76946.1 \\
 & & 10000 & 55771.5 & 317997.3 & 86320.7 & 55706.6 & 317968.5 & 86279.6 \\
  \cmidrule(lr){2-9}
 & \multirow{4}{*}{t7}
   & 500   & 2430.9 & 13369.0 & 11956.5 & 2393.0 & 13352.1 & 11939.0 \\
 & & 1000  & 4797.0 & 26364.6 & 21940.2 & 4752.8 & 26344.9 & 21919.5 \\
 & & 5000  & 23802.1 & 131415.3 & 35270.9 & 23743.4 & 131389.2 & 35233.9 \\
 & & 10000 & 47642.1 & 263206.3 & 47628.5 & 47577.2 & 263177.4 & 47581.2 \\
\hline 
\multirow{12}{*}{$d = 8$}
 & \multirow{4}{*}{Normal}
   & 500   & 1811.1 & 9426.0 & 7999.0 & 1760.5 & 9409.1 & 7981.3 \\
 & & 1000  & 3578.4 & 18809.8 & 14825.9 & 3519.5 & 18790.2 & 14804.7 \\
 & & 5000  & 17657.3 & 94077.2 & 18884.5 & 17579.1 & 94051.1 & 18840.7 \\
 & & 10000 & 35254.2 & 188017.9 & 35220.4 & 35167.7 & 187989.0 & 35163.9 \\
  \cmidrule(lr){2-9}
 & \multirow{4}{*}{t5}
   & 500   & 2871.9 & 16424.7 & 15881.0 & 2821.3 & 16407.9 & 15863.9 \\
 & & 1000  & 5721.2 & 32643.4 & 31008.7 & 5662.3 & 32623.8 & 30988.8 \\
 & & 5000  & 28417.1 & 162506.6 & 120038.4 & 28338.9 & 162480.6 & 120009.9 \\
 & & 10000 & 56694.0 & 323298.9 & 153982.5 & 56607.5 & 323270.1 & 153946.3 \\
  \cmidrule(lr){2-9}
 & \multirow{4}{*}{t7}
   & 500   & 2494.5 & 13527.3 & 12907.2 & 2443.9 & 13510.5 & 12890.1 \\
 & & 1000  & 4950.8 & 27136.5 & 24769.2 & 4891.9 & 27116.9 & 24749.0 \\
 & & 5000  & 24362.2 & 134338.2 & 62251.2 & 24284.0 & 134312.2 & 62218.4 \\
 & & 10000 & 48806.6 & 269491.9 & 69318.2 & 48720.1 & 269463.1 & 69274.0 \\
\hline
\end{tabular}
}
\caption{Average BIC and AIC of the 500 simulation runs for Scenario 1 with exogenous covariates with exponentially decaying covariance in a GARCH(2,1)-X model.} \label{tab:GARCH_ExpDecay_AIC}
\end{table}

\section{Real Data Analysis: Volatility of the SP500 based on economic and commodity indicators} \label{sec.real_data}

Modeling the volatility of the S\&P 500 index is of central importance in financial econometrics, with applications in risk management, asset pricing, portfolio optimization, and regulatory capital assessment. While traditional GARCH models capture time-varying volatility through its endogenous properties, they often fail to account for the influence of external economic and financial conditions. Incorporating such variables can improve financial planning, particularly during periods of market stress when endogenous dynamics alone are insufficient.

In this section, we examine whether incorporating a set of macro-financial indicators enhances volatility modeling of the S\&P 500. Specifically, we include as covariates the prices of NASDAQ and Dow Jones Industrial Average indices, as well as key commodities: crude oil, natural gas, gold, palladium, rice, steel, and wheat. These variables are economically motivated. The NASDAQ and DOW indices reflect broader market conditions and investor sentiment, and may serve as indicators of volatility in the S\&P 500. Energy and industrial commodities, such as crude oil and natural gas, influence input costs, inflation expectations, and economic uncertainty. Precious metals like gold and palladium are often considered safe-haven assets and tend to react strongly during times of market turmoil. Agricultural commodities, including rice and wheat, as well as industrial inputs like steel, provide information about supply chain pressures and global economic demand. By incorporating these heterogeneous indicators into the volatility modeling framework, we aim to capture both financial contagion and macroeconomic effects that may influence S\&P 500 volatility beyond its own endogenous behavior.

The dataset consists of daily closing prices of the exogenous covariates from January 2000 to May 2025. A GARCH(1,1)-X model was fitted under three configurations: using all available covariates (referred to as the full model), using only those covariates selected by the proposed variable selection method (selected model), and excluding all covariates (null model). The adjusted p-values obtained through the proposed selection procedure, following false discovery rate (FDR) correction, were as follows: NASDAQ (0.0464), Dow Jones (0.1428), crude oil (0.0464), gold (0.4759), natural gas (0.1300), palladium (0.1428), rice (0.3366), steel (0.0707), and wheat (0.1968). These results indicate that, among the variables considered, only the NASDAQ and crude oil prices are selected as relevant exogenous predictors of S\&P 500 volatility at 0.05 adjusted significance levels. Additionally, the adjusted p-value for steel is marginally significant, suggesting it may have some explanatory power, albeit weaker than the other two covariates.

Two versions of the selected model were evaluated: one including the NASDAQ and crude oil covariates, and another incorporating NASDAQ, crude oil, and steel prices. These were compared against both the null model (excluding all covariates) and the full model (including all covariates). The BIC and AIC values for the null model were 141704.6 and 141685.5, respectively. For the full model, the BIC and AIC were 141792.3 and 141715.9. The selected model excluding steel yielded a BIC of 141721.2 and an AIC of 141689.4, while the model including steel produced a BIC of 141729.7 and an AIC of 141691.5. Additionally, GARCH models with higher orders for the lag parameters $p$ and $q$ were tested but resulted in higher AIC and BIC values, indicating inferior performance in terms of model selection criteria.

The comparison of AIC and BIC values across the null, selected, and full GARCH-X models reveals important insights. The null model, which includes no exogenous covariates, achieves slightly lower BIC (141704.6) and AIC (141685.5) compared to the other models. The selected models incorporating NASDAQ and crude oil (with or without steel) follow very closely, and remain notably more parsimonious than the full model, which includes all covariates and performs worse under both criteria. 

The modest increase in AIC/BIC for the selected models—particularly the one including only NASDAQ and crude oil prices—indicates that these variables may still contain useful  information and improved interpretability, given that they are statistically significant in the model and the sample size is large and the penalty is very small compared to the size of the likelihood.
The model including steel prices shows slightly worse BIC and AIC than the version without it, reinforcing the interpretation of steel as only marginally informative.

\pagebreak

\section{Conclusion}\label{sec.conclusion}

This paper proposes a consistent variable selection procedure for GARCH-X models that exploits the connection between model checking and variable selection. By leveraging the false discovery rate (FDR) control method of Benjamini and Yekutieli, we construct a test-based selection rule that asymptotically recovers the true set of exogenous covariates influencing conditional volatility. Theoretical results establish the consistency of the procedure under mild conditions, including allowance for diminishing signal strength as the sample size grows.

Through extensive simulations across various model configurations, error distributions, and covariate dependence structures, we demonstrate the robustness and effectiveness of the method in finite samples. The real data application to S\&P 500 volatility illustrates the procedure applied to a real dataset. The proposed approach provides a framework for achieving model parsimony, while identifying essential external drivers of volatility and avoiding overfitting or unnecessary complexity. 

\backmatter





\bmhead{Acknowledgements}

The authors gratefully acknowledge support from the National Science Foundation under Grant No. NSF 2413081. This support made possible the research and development in this paper.












\begin{appendices}

\section{Assumptions}\label{secA1}

The following assumptions are necessary conditions for Theorem \ref{thm.consistency} and are extensively discussed in \cite{FrancqThieu2019}.
\begin{description}

\item[\textbf{A1:}] $(w_t, x_t')$ is a strictly stationary and ergodic process, and there exists $s > 0$ such that $\mathbb{E}|w_1|^s < \infty$ and $\mathbb{E}\|x_1\|^s < \infty$.

\item[\textbf{A2:}] $\mathbb{E}(w_t \mid \mathcal{F}_{t-1}) = 0$ and $\mathbb{E}(w_t^2 \mid \mathcal{F}_{t-1}) = 1$.

\item[\textbf{A3:}] $\vtheta \in \Theta$, and $\Theta$ is compact.

\item[\textbf{A4:}] For all $i \geq 1$ (and all $t \in \mathbb{Z}$), the support of the distribution of $w_{t-i} \mid \mathcal{F}_{t,i}$ is not included in $[0, \infty)$ or $(-\infty, 0]$ and contains at least three points.

\item[\textbf{A5:}] $\gamma < 0$ and $\sum_{j=1}^p \beta_j < 1$ for all $\vtheta \in \Theta$.

\item[\textbf{A6:}] There exists $s > 0$ such that $\mathbb{E}\sigma_t^s < \infty$ and $\mathbb{E}|\varepsilon_t|^s < \infty$.

\item[\textbf{A7:}] If $q > 0$, then $B_{\vtheta}(z)$ has no common root with $A_{\vtheta}^{+}(z)$ and $A_{\vtheta}^{-}(z)$; also $A_{\vtheta}^{+}(1) + A_{\vtheta}^{-}(1) \neq 0$ and $\alpha_{0q}^{+} + \alpha_{0q}^{-} + \beta_{0p} \neq 0$.

\item[\textbf{A8:}] If $c \in \mathbb{R}^r \setminus \{0\}$, then $c'x_1$ is not degenerate.

\item[\textbf{A9:}] $\mathcal{C} := \bigcup_{n=1}^{\infty} \left\{ \sqrt{n}(\vtheta - \vtheta): \vtheta \in \Theta \right\} = \prod_{i=1}^d C_i$, where $C_i = [0, \infty)$ if $\vtheta_{i} = 0$, and $C_i = \mathbb{R}$ otherwise.

\item[\textbf{A10:}] $\mathbb{E} w_t^4 < \infty$ in Cases A and B, and $\mathbb{E}|w_t|^{4+\nu} < \infty$ for some $\nu > 0$ in Cases C and D.

\item[\textbf{A11:}] For the power $\delta_0$ of the APARCH model (2): \\
In Case B: $\mathbb{E}|\varepsilon_t|^{2\delta_0 + \nu_1} < \infty$ and $\mathbb{E}\|x_t\|^{2 + \nu_1} < \infty$ for some $\nu_1 > 0$; \\
In Case D: $\mathbb{E}|\varepsilon_t|^{2\delta_0 + 8\delta_0/\nu} < \infty$ and $\mathbb{E}\|x_t\|^{2 + 8/\nu} < \infty$ with $\nu$ satisfying A10.

Here the top Lyapunov exponent \( \gamma \) associated with the APARCH-X model is defined as:
\[
\gamma := \lim_{t \to \infty} \frac{1}{t} \log \left\| C_{0t} C_{0,t-1} \cdots C_{01} \right\| \quad \text{a.s.}
\]

The matrix \( C_{0t} = [C_{0t}(i,j)] \) has the following non-zero elements:

\begin{align*}
& C_{0t}(1, 1) = \alpha_{01}^+ (\eta_t^+)^{\delta_0} + \alpha_{01}^- (\eta_t^-)^{\delta_0} + \beta_{01}, \\
& C_{0t}(1, j) = \beta_{0j} \quad \text{for } j = 2, \dots, p, \\
& C_{0t}(1, p + 2i - 3) = \alpha_{0i}^+ \quad \text{for } i = 2, \dots, q, \\
& C_{0t}(1, p + 2i - 2) = \alpha_{0i}^- \quad \text{for } i = 2, \dots, q, \\
& C_{0t}(i, i - 1) = 1 \quad \text{for } i = 2, \dots, p, \\
& C_{0t}(p + 1, 1) = (\eta_t^+)^{\delta_0}, \\
& C_{0t}(p + 2, 1) = (\eta_t^-)^{\delta_0}, \\
& C_{0t}(p + 2i - 1, p + 2i - 3) = 1 \quad \text{for } i = 2, \dots, q - 1, \\
& C_{0t}(p + 2i, p + 2i - 2) = 1 \quad \text{for } i = 2, \dots, q - 1.
\end{align*}

All other entries of \( C_{0t} \) are zero.

The polynomial \( B_\theta(z) \) corresponds to the autoregressive part of the conditional variance equation and is defined as:
$
B_\vartheta(z) = 1 - \sum_{j=1}^{q} \beta_j z^j,
$
where
   \( \beta_j \geq 0 \) are the coefficients associated with past conditional variances \( h_{t-j} \),
   \( q \) is the order of the GARCH component,
   \( z \in \mathbb{C} \) is a complex variable.

\end{description}




\end{appendices}


\clearpage

\bibliography{mybibfile}

\end{document}